\documentclass[11pt]{article}
 
\usepackage{amsmath,amsthm} 
\usepackage{amsfonts} 
\usepackage{amssymb}
\usepackage{graphicx} 

\usepackage[utf8]{inputenc}
\usepackage[english]{babel}
 
\usepackage[a4paper]{geometry}

\newcommand{\qpoch}[2]{(#1;q)_{#2}}
\newcommand{\qqpoch}[2]{(#1;q^2)_{#2}}

\newcommand{\Lam}{\Lambda} 
 
\newcommand{\al}{\alpha}
\newcommand{\be}{\beta}

\newtheorem{corollary}{Corollary}

\newtheorem{theorem}{Theorem}
\newtheorem{definition}{Definition}

\newtheorem{conjecture}{Conjecture}

\theoremstyle{remark}

\newcommand{\conj}{Assuming Conjecture 1 is true.\ }

\numberwithin{equation}{section}

\newcommand{\lf}{\mathcal{L}}
\newcommand{\bmm}{\mathbf{B}}
\newcommand{\lmm}{\mathbf{L}}
\newcommand{\limm}{\mathbf{L}^{-1}}
\newcommand{\uimm}{\mathbf{U}^{-1}}
\newcommand{\llmm}{(\mathbf{L}^{-1}\mathbf{L})}
\newcommand{\lbmm}{(\mathbf{L}^{-1}\mathbf{B})}
\newcommand{\umm}{\mathbf{U}}
\newcommand{\dmm}{\mathbf{D}}
\newcommand{\set}[1]{\left\{#1\right\}}

\newcommand{\R}{\mathcal{R}}

\newcommand{\Z}{\mathbb{Z}}

\newcommand{\mmod}{\mathcal{M}}

\newcommand{\vsub}{\mathcal{V}}

\newcommand{\nP}{\hat{P}}
\newcommand{\nQ}{\hat{Q}}

\newcommand{\st}{\,:\,}    % such that

\newcommand{\thmref}[1]{Theorem \ref{#1}}

\newcommand{\defref}[1]{Definition \ref{#1}}
\newcommand{\figref}[1]{Figure \ref{#1}}

\newcommand{\bra}[1]{\langle{#1}|}
\newcommand{\ket}[1]{|#1\rangle}
\newcommand{\bran}[1]{\langle{#1}}

\newcommand{\db}{\mathbf{d}}
 \newcommand{\eb}{\mathbf{e}}

\renewcommand{\natural}{\prime}
\renewcommand{\sharp}{{\circ}}
\renewcommand{\flat}{{\bullet}}
\begin{document}

%\linenumbers

\setcounter{page}{1}

\title{Bi-orthogonal Polynomials  and the Five parameter Asymmetric Simple Exclusion Process}

\author{R. Brak\thanks{rb1@unimelb.edu.au}\hspace{1ex}   and W. Moore   
    \vspace{0.15 in} \\
         School of Mathematics,\\
         The University of Melbourne\\
         Parkville,  Victoria 3052,\\
         Australia\\
 }

\maketitle

\begin{abstract} 

We apply the bi-moment determinant method to compute a representation of the matrix product algebra --  a quadratic algebra satisfied by the operators $\db$ and $\eb$ --  for the five parameter ($\alpha$, $\beta$, $\gamma$, $\delta$ and $q$)  Asymmetric Simple Exclusion Process. This method requires an $LDU$ decomposition of the ``bi-moment matrix''. The decomposition defines a new pair of  basis vectors sets, the `boundary basis'.  This basis is defined by the action of polynomials $\set{P_n}$ and $\set{Q_n}$ on the quantum oscillator basis (and its dual). Theses polynomials are orthogonal to themselves (ie.\ each satisfy a three term recurrence relation) and are orthogonal to each other  (with respect to the same linear functional defining the stationary state). Hence termed `bi-orthogonal'.  With respect to the boundary basis the bi-moment matrix is diagonal and the representation of  the operator $\db+\eb$ is tri-diagonal.  This tri-diagonal matrix defines another set of orthogonal polynomials very closely related to the the Askey-Wilson polynomials (they have the same moments).

% We extend previous results for the three parameter diffusion algebra of the Asymmetric Simple Exclusion Process to five parameters: $\alpha$, $\beta$, $\gamma$, $\delta$ and $q$. We find a pair of basis changes that leads to the $LDU$ decomposition of the `bi-moment' matrix. 

% Associated with this pair of bases are three sequences of orthogonal polynomials. The first pair of  orthogonal polynomials generate the new basis vectors  (the \emph{boundary} basis) by their action on the `boundary' vectors (written is the standard basis), whilst the third orthogonal polynomials are essentially the Askey-Wilson polynomials. All theses results are ultimately related to the $LDU$ decomposition of a  matrix. 
%In order to derive the $LDU$ decomposition we need to conjecture the recurrence relation the lower triangular matrix elements satisfy.

\end{abstract}

\vfill
\newpage

%============================================
\section{Introduction}\label{sec1}
%============================================

% This paper is an extension of that by Brak and Moore \cite{brakMore2015} who considered the three parameter   Asymmetric Simple Exclusion Process  (ASEP). Here we extend this to the full five parameter case. 

The ASEP is a continuous time Markov process defined by particles hopping along a line of  $L$ sites -- see \figref{fig:hop}.  Particles hop on  to the line on the left (resp. right) with rates $\alpha$  (resp. $\delta$), off at the right (resp.left) with rate $\beta$ (resp. $\gamma$) and   they hop to neighbouring sites to the left with  rate $q$ and rate one to the right with the constraint that only one particle can occupy a site. 
 \begin{figure}[ht]
\begin{center}
\includegraphics[width=12cm]{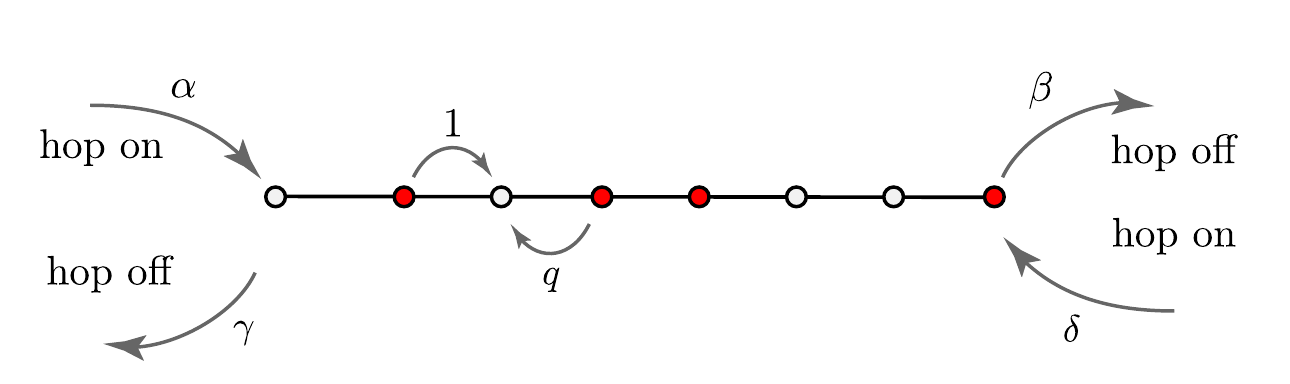}  
\caption{Five parameter ASEP hopping model}\label{fig:hop}
\end{center}
\end{figure}

The matrix product Ansatz \cite{derrida97} expresses the stationary distribution of a given state as an inner product on a certain product of matrices  $D$ and $E$ which satisfy the relation
\[
DE-qED-D-E=0
\]
and requires two vectors $\bra{W}$ and $\ket{V}$ which satisfy 
%\todo{confusing sentence}
\begin{subequations}\label{eq_bound}
\begin{align}
     (\beta D - \delta E -1 )\ket{V}&=0\,,\\
    \bra{W}(\alpha E - \gamma D -1)&=0\,.
\end{align}
\end{subequations}
We will call $\bra{W}$ and $ \ket{V}$   the \textbf{boundary vectors}.
The vectors $\bra{W}$ and $\ket{V}$ are used to define a linear functional which  maps any (non-commutative) polynomial, $p(D,E)$ in the matrices $D$ and $E$  to the set of (commuting) polynomials, $\Z[\alpha,\beta,\gamma,\delta,q]$,  via
\begin{equation}
\bra{W} p(D,E)\ket{V}\in \Z[\alpha,\beta,\gamma,\delta,q]\,.
\end{equation}

%The matrices arise as representations of the  ASEP diffusion  algebra \cite{derrida97,Isaev:2001aa}. 
The  representations  of the matrices $D$ and $E$ fall into three natural cases;
\begin{align*}
    \text{Two parameter:}&\qquad \alpha,\quad\beta ;\qquad q=\gamma=\delta=0\,, \\
    \text{Three parameter:}&\qquad q,\quad \alpha,\quad\beta ;\qquad \gamma=\delta=0\,, \\
        \text{Five parameter:}&\qquad q,\quad\alpha,\quad \beta,\quad \gamma,\quad \delta\,. \\
\end{align*}
The two parameter case is (algebraically) simple. The three parameter case has been studied   in  \cite{brakMore2015}. In this paper we apply the method introduced in \cite{brakMore2015} to the five parameter case. This generalisation is \emph{not} a simple extension of the three parameter case -- several  new difficulties appear. The more important ones are discussed in the Concluding Remarks section. 

For the five parameter case, new parameters are defined from the five hopping parameters, leading to the following change in the set of parameters
\begin{subequations}\label{abcd_def}
\begin{align}
a &= \frac{1}{2\alpha}(1-q-\alpha + \gamma + \sqrt{(1-q-\alpha+\gamma)^2 + 4\alpha \gamma})\,,\\
c &= \frac{1}{2\alpha}(1-q-\alpha + \gamma - \sqrt{(1-q-\alpha+\gamma)^2 + 4\alpha \gamma})\,,\\
b &= \frac{1}{2\beta}(1-q-\beta + \delta + \sqrt{(1-q-\beta+\delta)^2 + 4\beta \delta})\,,\\
d &= \frac{1}{2\beta}(1-q-\beta + \delta - \sqrt{(1-q-\beta+\delta)^2 + 4\beta \delta})\,.
\end{align}
\end{subequations}
This change is motivated by the parameters that occur in the Askey-Wilson polynomials \cite{gasper90} discussed further below.

Rather than using $D$ and $E$  the algebra is  simplified by working with the standard shifted variables,   
\begin{subequations}
\begin{align}\label{eq_shift}
	\db&=q'D-1\,,\\
	\eb&=q'E-1\,,
\end{align}	
\end{subequations}
where $q'=1-q$. In these variables the commutation relations of $\db$ and $\eb$ become \cite{Lazarescu:2014ab},
\begin{subequations}
\begin{align}
    \text{Two parameter:}&\qquad \db \eb=1\,, \\
    \text{Three and Five parameter:}&\qquad \db \eb-q\, \eb \db=q' \,.
    \end{align}
    \end{subequations}
 and \eqref{eq_bound} can be written  in the form, 
 \begin{subequations}
\begin{align}
(\db+ bd \eb -(b+d)\mathbf{1})\ket{V}&=0\,,\label{eq_vbdry}\\
\bra{W}(\eb+ ac\db-(a+c)\mathbf{1})&=0\,.\label{eq_wbdry}
\end{align}\label{eq_bdry}
\end{subequations}

Each matrix representation of $\db$ and $\eb$ is  associated with a basis for the vector space upon which the matrices act. The standard quantum oscillator basis is the set  $\set{\ket{n}\,:\, n\ge0} $.
If the linear operators $\db$ and $\eb$ in the respective cases are defined by their action on the  basis vectors $\ket{n}$ by:
\begin{align*}
    \text{Two parameter:}&\qquad \ket{n+1}=\eb  \ket{n}\,,\\
    &\qquad  \db \ket{n}=\ket{n-1},\quad \db\ket{0}=0\,. \\
    \text{Three and Five parameter:}&\qquad \ket{n+1}=\eb  \ket{n}\,,\\
    &\qquad \db \ket{n}=(1-q^n)\ket{n-1},\quad \db\ket{0}=0\,.
\end{align*}
then it is simple to show  that $\db \eb=1$ (two parameter) and $\db \eb-q\,\eb \db=q'$ (three and five parameter respec.). 
Thus the basis   $\set{\ket{n}}$, in conjunction with the action of $\db$ and $\eb$ above, gives the  standard, \cite{Littlewood:1933aa},  matrix representation for $\db$ and $\eb$ which satisfy the appropriate commutation relations. 
In this  representation the  matrix $\db+\eb$ is   tri-diagonal and for the three and five parameter cases gives a three term recurrence related to  $q$-Hermite polynomials \cite{UCHIYAMA:2003aa}.

To find the vector  $\ket{V}$ (respec.\ $\bra{W}$) there are (at least) two approaches. 
The first is to  express $\ket{V}$ (respec.\  $\bra{W}$) as a linear combination of the standard basis vectors ie.\ $\ket{V}=\sum_{n}a_n \ket{n}$  (respec.\  $\bra{W}=\sum_{n} b_n\bra{n}$), and then compute the  coefficient $a_n$ (respec.\  $b_n$). 
For example, in the three parameter case this leads to 
\begin{equation}
    \ket{V}=a_0\sum_{n\ge 0} \frac{\left(   q'/\beta -1\right)^n}{\prod_{k=1}^{n}(1-q^k)} \ket{n}\,. 
\end{equation}

The second approach (used in this paper) is to find a new pair of bases $\set{\ket{V_n} \st   n\ge0}$ and  $\set{\bra{W_n}\st   n\ge0}$ such at  $\ket{V}= \ket{V_0 } $ and $\bra{W}= \bra{W_0} $. 
This pair of bases are constructed such that the ``bimoment matrix'', $\bmm$, is diagonal. 
The  matrix elements of $\bmm$ are defined by a linear functional $\lf$ (see \defref{thm_cw}) via
\begin{equation}\label{eq_bmmdef}
    \bmm_{n,m}=\lf( \db^n \eb^m )\,.
\end{equation}
We will call $\set{\ket{V_n}}$ and $\set{\bra{W_n}}$ the \textbf{boundary basis}. 
This method of finding a basis (and hence representation) reduces to computing determinants and ultimately to finding an $LDU$ decomposition.

Representations of the $\db$ and $\eb$ matrices for the five parameter model can be found in \cite{UCHIYAMA:2003aa} (and references therein), with one of those representations reproduced in  \eqref{dematrices}. 
If the matrices associated with a given representation  have sufficiently simple structure (eg.\ bi- or tri-diagonal) then they can be usefully interpreted as transfer  matrices  for  lattice path models \cite{Brak2004vf}. This leads to   combinatorial methods for computing  the inner product (or linear functional, $\lf$). 
%The  matrix product  Ansatz was originally proposed and justified in \cite{derrida97} by showing it satisfied the stationary master equation of the  ASEP.  
 
The primary objective of this paper is to the find the change of basis  associated with the five parameter model representation obtained by Uchiyama et. al.\ \cite{UCHIYAMA:2003aa} where   the tri-diagonal matrix $\db+\eb$  gives a three term recurrence related to the Askey-Wilson polynomials.  
As will be shown this change of basis is affected by the action of  sequences of polynomials (with matrix argument) acting on the boundary vectors.

The Askey-Wilson polynomials play  a prominent role in the representation of the $\db$ and $\eb$ matrices. The polynomials also motivate the $a$, $b$, $c$ and $d$ choice of parameters (rather than the hopping rates) defined above and several other choices defined below. We thus briefly discuss the Askey-Wilson polynomials.

The Askey-Wilson polynomials \cite{askeywilson, gasperrahman, koekoekleskyswarttouw} are  `$q$-orthogonal' polynomials  with four parameters, $a,b,c$ and $d$ (and $q$). They are at the top of the Askey-scheme of $q$-orthogonal, one variable polynomials.
The basic hypergeometric functions, ${~}_r \phi_s$, give  a compact expression for the Askey-Wilson polynomial, $W_n(x) = W_n(x;a,b,c,d|q)$, $n>0$, which is given by 
\begin{equation}
W_n(x) = a^{-n}(ab,ac,ad;q)_n\,\, {}_4 \phi_3\left[
\begin{matrix}
    q^{-n},& q^{n-1} abcd, & ae^{i\theta},& a e^{-i\theta}\\ 
    ab,& ac,& ad 
\end{matrix} 
;q,q \right]
\end{equation}
with $x=\cos\theta$  and the  basic hypergeometric function is 
\begin{equation}
{~}_r \phi_s
\left[
\begin{matrix}a_1, & \hdots, & a_r\\
b_1, & \hdots, & b_s
\end{matrix}
;q,z \right] = 
\sum_{k=0}^\infty \frac{(a_1,\hdots,a_r;q)_k}{(b_1,\hdots, b_s,q;q)_k}
\left(  (-1)^k q^{\binom{k}{2}} \right)^{1+s-r}z^k
\end{equation}
where  the $q$-shifted factorial is  
\begin{equation}
(a_1,a_2,\hdots,a_s;q)_n = \prod_{r=1}^s \prod_{k=0}^{n-1} (1-a_r q^k).
\end{equation}
The Askey-Wilson polynomial satisfies a three-term recurrence relation 
\begin{equation}
A_n W_{n+1}(x) + B_n W(x) + C_n W_{n-1}(x) = 2x W_n(x)\,,\label{eq_awrecrel}
\end{equation}
with $W_0(x)=1$, $W_{-1}(x)=0$ and
\begin{multline}
\shoveright{ A_n  = \frac{1 - q^{n-1}abcd}{(1-q^{2n-1}abcd)(1-q^{2n}abcd)}\,,} 
\end{multline}
\begin{multline}
B_n = \frac{q^{n-1}}{(1-q^{2n-2}abcd)(1-q^{2n}abcd)}  [(1 + q^{2n-1}abcd)(qs+abcds')\\
  -q^{n-1}(1+q)abcd(s+qs')]\,,
\end{multline}
\begin{multline}
C_n = \frac{(1-q^n)(1-q^{n-1}ab)(1-q^{n-1}ac)(1-q^{n-1}ad)}{(1-q^{2n-1}abcd)}  \\     \times \frac{(1-q^{n-1}bc)(1-q^{n-1}bd)(1-q^{n-1}cd)}{(1-q^{2n-2}abcd)}
\end{multline}
and
\begin{equation}
s=a+b+c+d, \hspace{1.5cm} s'=a^{-1} + b^{-1} + c^{-1}+d^{-1}.
\end{equation}

Uchiyama et.\ al.\ \cite{UCHIYAMA:2003aa} found a representation of   $\db$  and $\eb$ related to the  Askey-Wilson polynomials. The matrices are tridiagonal and given by
\begin{subequations}

\begin{equation} \label{dematrices}
\db = \left(\begin{matrix}
d_0^\natural & d_0^\sharp & 0 & \hdots \\
d_0^\flat & d_1^\natural & d_1^\sharp  & \\
0 & d_1^\flat & d_2^\natural & \\
\vdots & & & \ddots
\end{matrix}\right)\qquad \text{and}\qquad
\eb = \left(\begin{matrix}
e_0^\natural & e_0^\sharp & 0 & \hdots \\
e_0^\flat & e_1^\natural & e_1^\sharp  & \\
0 & e_1^\flat & e_2^\natural & \\
\vdots & & & \ddots
\end{matrix}\right)
\end{equation}
where
\begin{align}
d_n^\natural =& \frac{q^{n-1}}{(1-q^{2n-2}abcd)(1-q^{2n}abcd)}\notag\\
&[bd(a+c) + (b+d)q - abcd(b+d)q^{n-1} - \{bd(a+c)+abcd(b+d)\}q^n\notag\\
&-bd(a+c)q^{n+1} + ab^2cd^2(a+c)q^{2n-1} + abcd(b+d)q^{2n}]\,,\\
e_n^\natural =& \frac{q^{n-1}}{(1-q^{2n-2}abcd)(1-q^{2n}abcd)}\notag\\
&[ac(b+d) + (a+c)q - abcd(a+c)q^{n-1} - \{ac(b+d)+abcd(a+c)\}q^n\notag\\
&-ac(b+d)q^{n+1} + a^2bc^2d(b+d)q^{2n-1} + abcd(a+c)q^{2n}]\,,\\
\mathcal{A}_n =&  \left[\frac{(1-abcdq^{n-1})(1-q^{n+1})(1-abq^n)(1-bcq^n)(1-adq^n)(1-cdq^n)}
{(1-abcdq^{2n-1})(1-abcdq^{2n})^2(1-abcdq^{2n+1})}\right]^\frac{1}{2}\,,
%\end{align}
\intertext{and}
%\begin{align}
d_n^\sharp &= \frac{1}{1-q^nac}\mathcal{A}_n\,, \qquad \qquad e_n^\sharp  = \frac{-q^nac}{1-q^nac}\mathcal{A}_n\,,\\
d_n^\flat &= \frac{-q^nbd}{1-q^nbd}\mathcal{A}_n\,, \qquad \qquad e_n^\flat  = \frac{1}{1-q^nbd}\mathcal{A}_n\,.
\end{align}\label{eq_awparam}
\end{subequations}
We have introduced the parameters, $d_n^\sharp$, $d_n^\flat$, $e_n^\sharp$ and $e_n^\flat$ as they will reoccur in computations in the rest of this paper.

% We never use this:
%If $p$ is a vector, written
% \begin{equation}
% \ket{p(x)} = (p_0(x),p_1(x),p_2(x),\hdots)^T  
% \end{equation}
% where
% \begin{equation}
% p_n(x) = W_n(x;a,b,c,d|q) \, ,
% \end{equation}
% then    $\ket{p(x)}$ is a right eigenvector of the matrix $\db+\eb$ with the eigenvalue $2x$.

%=============================================
\section{The Linear Functional}
%=============================================
\label{sec_mods}

%=============================================

In this section we set up the tensor algebra used to represent the ASEP \cite{Crampe:2014ab}. Let $\R$ be the ring of integer coefficient commutative polynomials, $\Z[\al,\be,\gamma,\delta,q]$ and $\mmod$   the  $\R$-module (or tensor algebra)
\begin{equation}\label{eq_mmod}
\mmod= \bigoplus_{n\ge0}  \vsub_2^{\otimes n}
\end{equation}
where $\vsub_2$ is a free rank two $\R$-module with generators $ \db$ and $\eb$. Here $\vsub_2^{\otimes 0}$ denotes the ring $\R$  of the module and $\vsub_2^{\otimes n}=\vsub_2\otimes\vsub_2\otimes\cdots\otimes\vsub_2$ ($n$ factors). The homogeneous submodule $\vsub_2^{\otimes n}$, of degree $n$, is generated by the  standard  monomial basis elements $e_{i_1} \otimes  e_{i_2} \cdots  \otimes e_{i_n} $ where $e_i\in\set{\db,\eb}$. 
For brevity  we will frequently omit the tensor product symbol, thus $\db^m \eb^n$ denotes  $\db^{ \otimes m}\otimes \eb^{ \otimes n}$ etc. 

We use the  five parameter version of the  original matrix Ansatz algebra equations of Derrida et.\ al.\ \cite{derrida97} as modified in \cite{corteel:2010rt}. The latter form  allows for arbitrary monomial pre- and post-factors ($u$ and $v$ in the equations below). 
The original algebra (in \cite{derrida97}) was stated   in terms of matrices and vectors. 
Here we give a slightly more abstract version by using  a  linear functional in terms of   $\db$ and $\eb$.
%=============================================
\begin{definition} \label{thm_cw}

Let $u,v$ be any monomial basis elements of $\mmod$. 
The \mbox{$\R$-module}  homomorphism  $\lf :\,\mmod\to \R$ is defined by the following equations:
\begin{subequations}\label{eq_cw}
\begin{align}
 \lf(u\otimes (\db \otimes  \eb - q\, \eb \otimes  \db -q' )\otimes v)&= 0\label{eq_cw_a}\,,\\
  \lf(u \otimes (\db +bd \eb -(b+d)\mathbf{1}) )&=0\label{eq_cw_b}\,,\\
 \lf((\eb +ac \db -(a+c) \mathbf{1} )\otimes v )&=0\,,\label{eq_cw_c}
\end{align}
\end{subequations}
where $a,b,c$ and $d$ are defined in \eqref{abcd_def},  $\lf(1)=1$  and extended linearly to other elements of $\mmod$.

\end{definition}
%=============================================

The matrix product Ansatz of  \cite{derrida97}  for the stationary state can now be (trivially) restated   using  the   linear functional $\lf$.
\begin{theorem}[ Derrida, Evans,  Hakim  and  Pasquier \cite{derrida97}]
The stationary state probability distribution, $f(\tau)$, of the five parameter ASEP for the system in state $\tau=(\tau_1,\cdots,\tau_L)$, is given by
\begin{align}
f(\tau) =&\frac{1}{Z_L}\, \lf\left(\prod_{i=1}^L(\tau_i \db+(1-\tau_i)\eb)\right)\, \label{eq_prob}
\intertext{where}
Z_L  =&\lf\bigl((\db+ \eb)^L\bigr)\,\label{norm_eq}
\end{align}
and $\tau_i=1$ if site $i$ is occupied and zero otherwise.
\end{theorem}
%

%=============================================
\section{Bi-Orthogonal Pair of Polynomial  Sequences}
%=============================================
\label{biops}

In this section we define a pair of polynomials sequences. These polynomials are then used to construct the boundary basis which leads to a matrix representation of $\db$ and $\eb$.

Consider the pair of sequences, 
\begin{equation}
 \set{P_n(\db)}_{n\ge0} \qquad \text{and }\qquad  \set{Q_n(\eb)}_{n\ge0} \label{eq_polypair}
\end{equation}
of  monic polynomials where $P_n$ and $Q_n$ are degree $n$. We wish to determine if it is possible to find such a pair which are  orthogonal with respect to $\lf$ (as defined in \defref{thm_cw}), that is  $\lf(P_n  Q_m)=\Lam_n\delta_{n,m}$, $\Lam_n\ne 0$? 
%These two sequences will then give us a basis and a dual basis for the representation associated with the Al-Salam-Chihara polynomials obtained in \cite{Sasamoto:1999aa,UCHIYAMA:2003aa}.  

In order to show such a pair of sequences does indeed exist we    
consider  the   infinite dimensional `bimoment matrix', $ \bmm $,  whose matrix elements are defined to be 
\begin{equation}\label{binom}
 \bmm_{n,m} =\lf(\db^{ n}\,    \eb^{ m})\,, \qquad n,m\ge 0\,.
\end{equation}
The bimoment  matrix elements  satisfy a pair of partial difference equations as  given in the following theorem. 
%=============================================
\begin{theorem}\label{bimompascal}
The bimoment matrix elements, \eqref{binom}, satisfy the recursions 
\begin{subequations}\label{derecs}
\begin{align}
\bmm_{i,j} &= (1-q^{i}) \bmm_{i-1,j-1} + (a+c)q^i \bmm_{i,j-1} - acq^i \bmm_{i+1,j-1} \label{erec}\,,\\
\bmm_{i,j} &= (1-q^{j}) \bmm_{i-1,j-1} + (b+d)q^j \bmm_{i-1,j} - bdq^j \bmm_{i-1,j+1} \label{drec}\,,
\end{align}	
\end{subequations}
for $i,j>0$ with boundary values $\bmm_{0,j} \textrm{ and } \bmm_{i,0}$, $i,j\ge0$ satisfying
\begin{subequations}\label{boundrecs}
\begin{align}
\bmm_{i,0} &= \frac{((b+d - bd(a+c)q^{i-1}) \bmm_{i-1,0} - bd(1-q^{i-1})\bmm_{i-2,0}}{1-abcdq^{i-1}} \label{lbound} \,,\\ 	
\bmm_{0,j} &= \frac{((a+c - ac(b+d)q^{j-1}) \bmm_{0,j-1} - ac(1-q^{j-1})\bmm_{0,j-2}}{1-abcdq^{j-1}} \label{ubound}
\end{align}
and $\bmm_{0,0}$ =1.  
\end{subequations}
\end{theorem}
Note, the bimoment matrix elements satisfy \emph{both} \eqref{erec} and \eqref{drec}, however, to generate the matrix elements it is sufficient to use only one of \eqref{erec} or \eqref{drec} together with both boundary recurrences. The reason both \eqref{erec} and \eqref{drec} are stated is that it makes explicit a symmetry of the matrix which we will use below.
%In fact \eqref{erec} with $j\geq 0$ and $i > 0$ and \eqref{lbound} are sufficient to calculate the bimoment matrix and equivalently for equations \eqref{drec} and \eqref{ubound}.
%==============
\begin{proof}
The idea of the proof is similar to \cite{brakMore2015}. However, eliminating an $\eb$ (resp. $\db$) from the left (resp. right) side of is more complicated due to the more complicated boundary vector equations \eqref{eq_bdry}. Using \eqref{eq_cw_c} (resp. \eqref{eq_cw_b}), an $\eb$ (resp. $\db$) on the left (resp. right) side of a monomial can be removed giving,
\begin{align}
\lf (\eb \db^n\eb^{m-1}) &=  (a+c) \lf (\db^n\eb^{m-1}) - ac \lf(\db^{n+1}\eb^{m-1})\,,\label{left_e}\\
\lf (\db^{n-1}\eb^m\db) &=  (b+d) \lf (\db^{n-1}\eb^m) - bd \lf(\db^{n-1}\eb^{m+1})\,.\label{right_d}
\end{align}
Similar to the proof in \cite{brakMore2015}, commuting an $\eb$ all the way to the left and then eliminating the $\eb$ on the left using \eqref{left_e} gives the following recurrence \eqref{erec},
\begin{align}
\lf(\db^n \eb^m) &=(1-q^n) \lf(\db^{n-1} \eb^{m-1}) + q^{n}\lf(\eb \db^n\eb^{m-1}) \notag\\ 	
	    &=(1-q^n) \lf(\db^{n-1} \eb^{m-1}) + (a+c)q^n \lf(\db^n\eb^{m-1})\notag \\
	    &\qquad - ac q^n \lf(\db^{n+1}\eb^{m-1})\,.
\end{align}
Commuting a $\db$ all the way to the right and then eliminating it using \eqref{right_d} gives the following recurrence \eqref{drec},
%\todo{check equation refs}
%
\begin{align}
\lf(\db^n \eb^m) &=(1-q^m) \lf(\db^{n-1} \eb^{m-1}) + q^{m}\lf(\db^{n-1}\eb^m \db) \notag\\ 	
	    &=(1-q^m) \lf(\db^{n-1} \eb^{m-1}) + (b+d)q^m \lf(\db^{n-1}\eb^m)\notag\\
	    & \qquad - bd q^m \lf(\db^{n-1}\eb^{m+1})\,.
\end{align}
By rearranging \eqref{eq_cw_b} and \eqref{eq_cw_c}, an $\eb$ can be eliminated from the left and a $\db$ eliminated from the right, as distinct from the 3-parameter case, giving,
\begin{align}
\lf (\db^n\eb) &= \frac{1}{bd} \left((b+d)\lf (\db^n) - \lf(\db^{n+1})\right)\,,\\
\lf (\db\eb^m) &= \frac{1}{ac} \left((a+c)\lf (\eb^m) - \lf(\eb^{m+1})\right)\,. 
\end{align}
Finally, using the recursions \eqref{erec} and \eqref{drec} for $m=1$ and $n=1$ gives the following expression in terms of the boundary values,	
\begin{align}
\lf(\db^n \eb) &=(1-q^n) \lf(\db^{n-1}) + (a+c)q^n\lf(\db^n) - ac q^n \lf(\db^{n+1})\,,\\
\lf(\db \eb^m) &=(1-q^m) \lf(\eb^{m-1}) + (b+d)q^m\lf(\eb^m) - bd q^m \lf(\eb^{m+1})\,. 
\end{align}
Combining these results gives a recurrence for the boundary value terms.	
\end{proof}
The existence of the pair of polynomial sequences  \eqref{eq_polypair}   requires that the determinant of the $(n+1)\times (n+1) $ sub-matrix  
$$\bmm^{(n)}=(\bmm_{i,j})_{0\le i,j\le n}$$ 
be non-zero for all $n\ge 0$ (see \cite{brakMore2015} for further details).

In the case of the three parameter model the corresponding determinant was evaluated using theorems from \cite{Krattenthaler:2002aa} and \cite{A.R.-Moghaddamfar:2008aa}. In this five parameter case we have been unable to find any similar theorems and thus attempted an $LDU$ decomposition directly.

For small values of $n$ the determinant,  $\det{\bmm^{(n)}}$ can be found by computer by iterating the recurrence relations \eqref{erec} to construct $\bmm$. 
From these values a product form for $\det{\bmm^{(n)}}$ (stated below in \eqref{eq_bidet}) can be conjectured. 
It is similarly possible to conjecture the $LDU$ decomposition of  $\bmm$, that is, find   upper and lower triangular matrices, $\umm$ and $\lmm$ respectively, such that 
\begin{equation}
    \bmm  = \lmm \mathbf{D} \umm \label{eq_lduf}
\end{equation}
(with $\mathbf{D}$ diagonal). The product of the first $n$ diagonal elements of  $\mathbf{D}$ then gives the determinants, $\det{\bmm^{(n)}}$.  
These small $n$ computations lead us to define a lower triangular matrix $\lmm$ via a recurrence relation for the matrix elements $\lmm_{i,j}$ given by
\begin{align} \label{L}
\lmm_{i,j} &=
\begin{cases}
 \lmm_{i-1,j-1} + d_j^{\natural} \lmm_{i-1,j} - bdq^jg_j \lmm_{i-1,j+1} \qquad &\text{for $i,j \geq 0$ and $i\geq j$\,,} \\
1 \qquad &\text{for $i=j=0$\,,} \\
0 \qquad &\text{otherwise\,,}
\end{cases}
\end{align} 
where 
\begin{align}
g_j &= \frac{(1-abcdq^{j-1})(1-q^{j+1})(1-abq^j)(1-bcq^j)(1-adq^j)(1-cdq^j)}{(1-abcdq^{2j-1})(1-abcdq^{2j})^2(1-abcdq^{2j+1})}\,. \label{eq_gi}
\end{align}	
For small values of $n$ this matrix (and a similar one for $\umm$) give  the $LDU$ decomposition of $\bmm$, but unfortunately we have not been able to prove the decomposition for arbitrary $n$. Thus we make the following conjecture.
\begin{conjecture} \label{eq_conj}
The bimoment matrix, \eqref{binom}, has an LDU decomposition with lower triangular matrix, $\lmm$, given by \eqref{L}.
\end{conjecture}
We make two remarks: first, one of the final results of this conjecture is a representation for $\db$ and $\eb$. Having obtained a candidate representation it is then straightforward to verify that is   a representation by substituting back into the defining algebra -- \defref{thm_cw}. This has been done and hence the representation verified. Assuming the logic of the calculation can be reversed, that would prove the conjecture. However, one of the primary aims of this paper is to \emph{derive} the representation and thus from this perspective is more satisfactory if the conjecture be proved directly.

Secondly, it is only necessary to conjecture $\lmm$ as, assuming \eqref{L} is valid, we can compute the corresponding recurrence relations for  the upper triangular matrix $\umm$ using a symmetry of $\bmm$. The bimoment matrix is invariant when taking the transpose and performing the substitutions $a \leftrightarrow b, c \leftrightarrow d $ or $a \leftrightarrow d, b \leftrightarrow c $. Under this action, the equations \eqref{erec} and \eqref{drec} swap as do the equations \eqref{lbound} and \eqref{ubound}. Thus the upper triangular matrix can be obtained from the lower triangular matrix by performing these substitutions.

\begin{corollary} \label{U}
\conj  The the upper triangular matrix elements $\umm_{i,j}$ of the $LDU$ decomposition of the bimoment matrix are given by the recurrence relation
\begin{align}
\umm_{i,j} &=
\begin{cases}
 \umm_{i-1,j-1} + e_i^{\natural} \umm_{i,j-1} -acq^ig_i \umm_{i+1,j-1}  \qquad &\text{for $i,j \geq 0$ and $i\leq j$\,,} \\
1 \qquad &\text{for $i=j=0$\,,} \\
0 \qquad &\text{otherwise\,.}
\end{cases}\label{eq_umat}
\end{align}
\end{corollary}

The diagonal matrix $\dmm$ of the $LDU$ decomposition can be calculated using the inverse of the $\lmm$ matrix. 

\begin{corollary}\label{tm_trirecrel}
\conj  Let $\lmm$ be the lower triangular matrix of the $LDU$ decomposition  of the bimoment matrix $\bmm$. Then the elements $\limm_{i,j}$ of the inverse of $\lmm$   satisfy  
\begin{align} \label{linvrec}
\limm_{i,j} = 
\begin{cases}
\limm_{i-1,j-1} - d_{i-1}^\natural \limm_{i-1,j} +bdq^{i-2}g_{i-2} \limm_{i-2,j} \qquad &\text{for $i,j \geq 0$ and $i\geq j$,} \\
1 \qquad &\text{for $i=j=0$,} \\
0 \qquad &\text{otherwise.}
\end{cases}
\end{align}	
Let $\umm$ be the upper triangular matrix of the $LDU$ decomposition  of the bimoment matrix $\bmm$. Then the elements $\uimm_{i,j}$ of the inverse of $\umm$  satisfy  
\begin{align} \label{uinvrec}
\uimm_{i,j} = 
\begin{cases}
\uimm_{i-1,j-1} - e_{j-1}^\natural \uimm_{i,j-1} +acq^{j-2}g_{j-2} \uimm_{i,j-2} \qquad &\text{for $i,j \geq 0$ and $i\leq j$,} \\
1 \qquad &\text{for $i=j=0$,} \\
0 \qquad &\text{otherwise.}
\end{cases}
\end{align}	
\end{corollary}
\begin{proof}
We will show that $\limm \lmm = \mathbf{1}$.  
Thus, substituting \eqref{linvrec} and then using \eqref{L} gives,
\begin{align}
\llmm_{i,j} & = \sum_k \lmm_{i,k}^{-1} \lmm_{k,j}\notag\\
&= \llmm_{i-1,j-1} + (d^{\natural}_j- d^{\natural}_{i-1}) \llmm_{i-1,j}\notag\\ 
&\qquad - bdq^jg_j \llmm_{i-1,j+1} +bdq^{i-2}g_{i-2} \llmm_{i-2,j} \label{llinvrec} \, . 
\end{align}
All entries above the main diagonal are zero since the matrix is lower triangular. On the diagonal \eqref{llinvrec} gives,
\begin{equation}
\llmm_{i,i}  = \llmm_{i-1,i-1} = \llmm_{0,0} = 1\,.
\end{equation}
 All that needs to be shown is that all the other diagonals contain only zeros. The recurrence \eqref{llinvrec} gives a matrix element in terms of elements from its own diagonal and the two diagonals above it. The only non-zero elements above the first off-diagonal are in the main diagonal. Therefore,
\begin{equation} 
\llmm_{i,i-1} = \llmm_{i-1,i-2} = \llmm_{0,-1} = 0\,.
\end{equation}
 Similarly for the second off-diagonal,
 \begin{equation}
\llmm_{i,i-2}  = \llmm_{i-1,i-3} = \llmm_{1,-1} = 0\,.
\end{equation}
Since the first two off-diagonals are zero, we get for $c>0$
\begin{equation}
\llmm_{i,i-c}  = \llmm_{i-1,i-1-c} = \llmm_{c-1,-1}=0\,.
\end{equation}
 The proof for $\uimm$ follows similarly.
\end{proof}
%=============================================
We can now compute the elements of the diagonal matrix $\dmm$.
\begin{theorem} \conj\
The diagonal matrix elements $\dmm_{n}$ of the  matrix $\dmm$ of \eqref{eq_lduf} satisfy a first order recurrence relation  giving
\begin{equation}\label{eq_bmmdiage}
\dmm_{n} =\prod_{i=0}^{n-1} g_{i} 
\end{equation}
for $n \geq 1$, where $g_i$ is given by \eqref{eq_gi} and   $\dmm_0 = 1$.
\end{theorem}

\begin{proof}
This proof follows similarly to the proof of $\limm$. Assuming the conjecture is true, $\lmm^{-1} \bmm$ is an upper triangular matrix.
% \begin{equation}
% \lbmm_{n,m} &= \sum_{i=0}^n \limm_{n,i}\bmm_{i,m}
% \end{equation} 
Using \eqref{linvrec}, the recurrence for $\limm$, gives, 
%&=\sum_{i=0}^n( \limm_{n-1,i-1} - d_{n-1}^\natural \limm_{n-1,i} +bdq^{n-2}g_{n-2} \limm_{n-2,i})\bmm_{i,m}\\
\begin{multline}
\lbmm_{n,m}  = \sum_{i=0}^n \limm_{n,i}\bmm_{i,m}
 =\sum_{i=0}^n \limm_{n-1,i-1}\bmm_{i,m} - d_{n-1}^\natural \lbmm_{n-1,m}\\
 +bdq^{n-2}g_{n-2} \lbmm_{n-2,m}\,.
\end{multline}
Using \eqref{drec}, the recurrence for the bimoment matrix, gives   
 \begin{multline}
\lbmm_{n,m}  =(1-q^{m}) \lbmm_{n-1,m-1} + ((b+d)q^m- d_{n-1}^\natural) \lbmm_{n-1,m}\\
- bdq^m g_m \lbmm_{n-1,m+1}  +bdq^{n-2}g_{n-2} \lbmm_{n-2,m} \,.
\end{multline} 
with $n=m+2$ and the fact that the matrix product is upper triangular gives the stated result.
\end{proof}

The value of the determinant, $\det\,\bmm^{(n)}$, is simple to calculate from the   $LDU$-decomposition of the bimoment matrix it being the product of the elements of the diagonal matrix.

\begin{theorem}\label{detthm} \conj\ 
Let $\bmm^{(n)}=(\bmm_{i,j} )_{0\le  i,j\le n}$ be the  truncated $(n+1)\times(n+1)$ bimoment matrix whose elements are defined by \thmref{bimompascal}.  Then 
%\todo[inline]{Rewrite in terms of Pochammer}
\begin{equation}\label{eq_bidet}
	\det\,\bmm^{(n)}=\\
	\prod_{i=1}^n \frac{\qpoch{abcd/q,q,ab,bc,ad,cd}{i}}{\qqpoch{abcd/q,abcd,abcd,abcdq}{i}}\,.
\end{equation}
\end{theorem} 

We now use the bimoment matrix to show the existence and uniqueness of the  polynomials sequences \eqref{eq_polypair}.  For $n,m\ge 0$  we require the bi-orthogonality condition
\begin{equation}\label{bi-orthog}
\lf(P_n(\db)\,  Q_m(\eb))=\Lam_n\delta_{n,m}
\end{equation}
where $\Lam_n$ is  a sequence of non-zero normalisation factors  determined by $\lf$ and the monic constraint. 

If this bi-orthogonality is translated into  the form of the original matrix product Ansatz, then the equation  is asking the question: Does there exist polynomials  $P_n(\db)$ and $Q_m(\eb)$  in the matrices $\db$ and $\eb$ such that
\begin{equation}\label{eq_maxbiorth}
	\bra{W}P_n(\db)\,Q_m(\eb)\ket{V}=\Lam_n\delta_{n,m}\, 
\end{equation}
for the boundary vectors  $\ket{V}$ and  $\bra{W}$ satisfying  \eqref{eq_bdry}.
% \begin{align}
%     (\db+bd\eb-(b+d)\mathbf{1})\ket{V}=0
% \intertext{and}
%     \bra{W}(\eb+ac\db-(a+c)\mathbf{1})=0\,.
% \end{align}
% $$$$ 
% and   
% $$\bra{W}(\eb+ac\db-(a+c)\mathbf{1})=0\,?$$ 
If the sequences exist we get a new pair of basis vectors   $ \ket{\hat{V}_n}_{n\ge0}$  and  their orthonormal  (with respect to $\lf$) duals  $ \bra{\hat{W}_n}_{n\ge0}$,  given by 
\begin{equation}\label{eq_orthobasis}
	\bra{\hat{W}_n}=\bra{W}P_n(\db)\frac{1}{\sqrt{\Lam_n}}\qquad\text{and}\qquad\ket{\hat{V}_n}=\frac{1}{\sqrt{\Lam_n}}Q_n(\eb)\ket{V}\,,
\end{equation}
where $\ket{\hat{V}_0}=\ket{ V }$ and $\bra{\hat{W}_0}=\bra{W }$. We normalise so that $ \bran{W }\ket{ V }=1$. 
%\todo{check V and W normalised} 
From these basis vectors,
%and since the identity matrix is $\mathbf{1}=\sum_{n\ge 0}\ket{\hat{V}_n} \bra{\hat{W}_n}$,  
we get matrix representations for $\db$ and $\eb$ by computing 
\begin{equation}\label{eq_bdrybasis}
	\db_{n,m}=\bra{\hat{W}_n} \db \ket{\hat{V}_m}\qquad \text{and}\qquad \eb_{n,m}=\bra{\hat{W}_n} \eb\ket{\hat{V}_m}\,.
\end{equation}

Returning to the question of the existence of bi-orthogonal polynomials  we  have the following theorem   stating a unique   pair of sequences exists.
%=============================================
\begin{theorem} \label{thm_othog}\conj\ 
 Let   $\set{P_n(\db)}_{n\ge0}$ and $\set{Q_n(\eb)}_{n\ge0}$ be a pair of sequences of monic polynomials   satisfying 
\begin{equation}
    \lf(P_n  Q_m)=\Lam_n\delta_{n,m}
\end{equation}
where  the linear functional $\lf$ is defined by  equations \eqref{eq_cw}. 
Then $\set{P_n}_{n\ge0}$ and $\set{Q_n}_{n\ge0}$ exist and are unique with
\begin{equation}\label{eq_lamn}
	\Lam_n= \dmm_n
	%\prod_{j=1}^n  \frac{(1-abcdq^{j-1})(1-q^{j+1})(1-abq^j)(1-bcq^j)(1-adq^j)(1-cdq^j)}{(1-abcdq^{2j-1})(1-abcdq^{2j})^2(1-abcdq^{2j+1})}   \,.
\end{equation}
for $n\geq0$ %and $\Lam_0=1$.
% and $g_i$ is given by \eqref{eq_gi}.
\end{theorem}
%=============================================

The proof is exactly the same as that which appears in \cite{brakMore2015} (except for the different value of $\Lam_n$) and thus we omit it.

To find the explicit form of the polynomials we need to evaluate  two determinants (see \cite{brakMore2015} for their derivation),
 \begin{equation}\label{expolyp}
P_n(\db)=\frac{1}{\det\,\bmm^{(n-1)}}\det\left(\begin{matrix}
\bmm_{0,0} & \bmm_{0,1} & \hdots  &\bmm_{0,n-1}  &1\\
\bmm_{1,0} & \bmm_{1,1} & \hdots  & \bmm_{1,n-1}& \db\\
\vdots & \vdots &\ddots & \vdots & \vdots\\
\bmm_{n-1,0} & \bmm_{n-1,1} &  \hdots  &\bmm_{n-1,n-1}& \db^{n-1}\\ 
\bmm_{n,0} & \bmm_{n,1} &  \hdots  &\bmm_{n,n-1}& \db^n 
\end{matrix}\right) \,,
\end{equation}
and
\begin{equation}\label{expolyq} 
Q_n(\eb)=\frac{1}{\det\,\bmm^{(n-1)}}\det\left(\begin{matrix}
\bmm_{0,0} & \bmm_{0,1} & \hdots & \bmm_{0,n-1} &\bmm_{0,n}\\
\bmm_{1,0} & \bmm_{1,1} & \hdots & \bmm_{1,n-1} & \bmm_{1,n}\\
\vdots & \vdots &\ddots & \vdots & \vdots\\
\bmm_{n-1,0} & \bmm_{n-1,1} &  \hdots &\bmm_{n-1,n-1}  &\bmm_{n-1,n}\\
1 & \eb &\hdots & \eb^{n-1} & \eb^n 
\end{matrix}\right)\,.
\end{equation}	
The two determinants can be evaluated  by  $LDU$ decomposition  of the two matrices leading to the following result.
%=============================================
\begin{theorem}  \label{lem_polrec} \conj\ 
The  pair of sequences of monic polynomials   $\set{P_n(\db)}_{n\ge0}$ and $\set{Q_n(\eb)}_{n\ge0}$   satisfy  
\begin{subequations}\label{eq_polexpn}
\begin{align}
	\db^n &= \sum_{k=0}^{n} \lmm_{n,k} P_k(\db)\label{eq_polexpnp}\,,\\
	\eb^n &= \sum_{k=0}^{n} Q_k(\eb)\umm_{k,n} \label{eq_polexpnq}\,,
\end{align}
\end{subequations}
where $\lmm_{n,k}$ and $\umm_{k,n}$ are the matrix elements of the lower triangular matrix $\lmm$ and  upper triangular matrix $\umm$ given by \eqref{L} and \eqref{eq_umat} respectively.
\end{theorem}
%=============================================

\begin{proof} Since the two matrices are very similar to the bimoment matrix thus, once the LDU decomposition of the the bimoment matrix is know that for \eqref{expolyp} and \eqref{expolyq} are readily obtained. Thus, assuming Conjecture 1, we get the following:
\begin{equation}\label{lduexpolyp}
P_n(\db)=\det\left(\begin{matrix}
\lmm_{0,0}		& 0 			&\hdots	& 0				& 1			\\
\lmm_{1,0}		& \lmm_{1,1} 	&\hdots	& 0				& \db			\\
\vdots			& \vdots 		&\ddots	& \vdots			& \vdots		\\
\lmm_{n-1,0}	& \lmm_{n-1,1}	&\hdots	& \lmm_{n-1,n-1}	& \db^{n-1}	\\
\lmm_{n,0}		& \lmm_{n,1}	&\hdots	& \lmm_{n,n-1}	& \db^n 
\end{matrix}\right)
\end{equation}
and
\begin{equation}\label{lduexpolyq} 
Q_n(\eb)=\det\left(\begin{matrix}
\umm_{0,0}	& \umm_{0,1}	& \hdots	& \umm_{0,n-1}	& \umm_{0,n}		\\
0			& \umm_{1,1}	& \hdots	& \umm_{1,n-1}	& \umm_{1,n}		\\
\vdots		& \vdots		& \ddots	& \vdots			& \vdots			\\
0			& 0				& \hdots	& \umm_{n-1,n-1}	& \umm_{n-1,n}	\\
1			& \eb				& \hdots	& \eb^{n-1}		& \eb^n 
\end{matrix}\right)\,.
\end{equation}

Expanding  \eqref{lduexpolyp} using  the bottom row leaves a sub-matrix determinant which reduces down to a $k \times k$ determinant of the same form   as \eqref{lduexpolyp}  but with $n=k$ and hence is   $P_k(d)$. Thus we get \eqref{eq_polexpnp}. Similarly for  \eqref{eq_polexpnq}.
\end{proof}

\begin{corollary} \label{cor_pold} \conj\ 
The  pair of sequences of monic polynomials   $\set{P_n(\db)}_{n\ge0}$ and $\set{Q_n(\eb)}_{n\ge0}$   can be expressed as 
\begin{subequations}\label{eq_explicitpol}
\begin{align}
	P_n(\db) &= \sum_{k=0}^{n} \limm_{n,k} \db^k \label{eq_expolp}\,,\\
	Q_n(\eb) &= \sum_{k=0}^{n} \eb^k \uimm_{k,n} \label{eq_expolq}\,,
\end{align}
\end{subequations}
where $\limm_{n,k}$ and $\uimm_{k,n}$ are the matrix elements of the inverse lower triangular $\limm$ and  inverse upper triangular $\uimm$    given by \eqref{linvrec} and \eqref{uinvrec} respectively.
\end{corollary}

We now use \eqref{eq_expolp} and \eqref{eq_expolq} to find a recursion formulation for $P_n$ and $Q_n$.
%=============================================
\begin{theorem}\label{tm_biorthpoys}\conj\ 
The  pair of sequences of monic polynomials $\set{P_n(\db)}_{n\ge0}$ and $\set{Q_n(\eb)}_{n\ge0}$ are given by the three-term recurrences 
\begin{align}
	\db P_{n}(\db) &= P_{n+1}(\db) + d_n^\natural\, P_{n}(\db) - bdq^{n-1}g_{n-1}\, P_{n-1}(\db)
\intertext{and}
	\eb Q_{n}(\eb) &= Q_{n+1}(\eb) + e_n^\natural\, Q_{n}(\eb) - acq^{n-1}g_{n-1}\, Q_{n-1}(\eb)   \label{eq_eqrec}
\end{align}
with $P_0=Q_0=1$, $P_{-1} = Q_{-1} = 0$ and the coefficients given by \eqref{eq_awparam}.
\end{theorem}
%=============================================
 \begin{proof}
We prove the recurrence relation of the bi-orthogonal polynomials by using the recurrence relations \eqref{linvrec} and \eqref{uinvrec} satisfied by the inverses of the upper and lower triangular matrix elements respectively.
Multiplying \eqref{eq_expolp} by $\db$ we get 
\begin{align}
\db P_n(\db) &= \sum_k \limm_{n,k} \db^{k+1}\\
&= \sum_k (\lmm_{n+1,k+1}^{-1} + d_{n}^\natural \limm_{n,k+1} - bdq^{n-1}g_{n-1} \limm_{n-1,k+1})\db^{k+1}\\
&= P_{n+1}(\db) + d_{n}^\natural P_n(\db) - bdq^{n-1}g_{n-1} P_{n-1}(\db)\,.
\end{align}
The proof of \eqref{eq_eqrec} follows similarly. 
\end{proof}
%=============================================

% %=============================================
  \section{Matrix representation and the boundary basis}  
% %=============================================

The recurrence relations for $P_n$ and $Q_n$  in \thmref{tm_biorthpoys}  can be used   to compute the following two moments (which lead to a matrix representation of $\db$ and $\eb$).  This gives the following theorem.

%=============================================
\begin{theorem} \label{thm_fmom}\conj\ 
Let $P_n$ and $Q_n$ be the polynomials of  \thmref{tm_biorthpoys}. The two first moments  
\begin{subequations}
\begin{align}
X_{n,m}&=\lf(P_n\, \db\, Q_m), \\
Y_{n,m}&=\lf(P_n \,\eb\, Q_m), 
\end{align}	
\end{subequations}
for $n,m\ge0$, are given by
\begin{subequations}
\begin{align}\label{fmmat}
X_{n,m}&= \Lam_{n+1} \delta_{n+1,m}+ d_n^{\natural} \Lam_{n}\delta_{n,m} - bdq^{n-1} \Lam_{n} \delta_{n-1,m}\,,  
 \\
Y_{n,m}&= \Lam_{m+1} \delta_{n,m+1}+ e_m^{\natural} \Lam_{m}\delta_{n,m} - acq^{m-1} \Lam_{m} \delta_{n,m-1}\, ,  
\end{align}	
\end{subequations}
for $n,m\ge 0$. 
\end{theorem}

To obtain a representation we need to use the orthonormal (non-monic) versions of the polynomials:
\begin{equation}%\label{eq_orthobasis}
	\nP_n(\db) = P_n(\db)\frac{1}{\sqrt{\Lam_n}}\qquad\text{and}\qquad \nQ_m(\eb)  =\frac{1}{\sqrt{\Lam_n}}Q_n(\eb) \,.
\end{equation}

%=============================================
 \begin{theorem} \label{tm_matprod}\conj\ 
 	 The  matrices $\db$ and $\eb$ with matrix elements
\begin{subequations}\label{eq_matrefdssd}
\begin{align}
	\db_{n,m}&=\lf(\nP_n \,\db\, \nQ_m)=X_{n,m}/\sqrt{\Lam_n\Lam_m}\,, \\
	\eb_{n,m}&=\lf(\nP_n\, \eb \,\nQ_m)=Y_{n,m}/\sqrt{\Lam_n\Lam_m}\,,, 
\end{align}
\end{subequations}
for $n,m\ge0$, give a matrix representation of  \eqref{eq_cw}.
\end{theorem}
 %=============================================
The theorem is proved by direct verification that the matrices \eqref{eq_matrefdssd} satisfy the quotient relation $\mathbf{\db  \eb} - q\, \mathbf{\eb \db} = q'\mathbf{1}$.

Using \eqref{fmmat} we see that $\db$ and $\eb$ have a  tri-diagonal structure
\begin{subequations}\label{eq_matrep}
\begin{equation}\label{eq_matrepd}
\db=\left(\begin{matrix}
d_0^\natural & \sqrt{g_0} & 0 & \hdots \\
-bd\sqrt{g_0} & d_1^\natural & \sqrt{g_1} & \hdots \\
0 & -bdq\sqrt{g_1} & d_2^\natural &  \hdots \\ 
\vdots & \vdots & \vdots & \ddots\\
\end{matrix}\right)
\end{equation}
and
\begin{equation}\label{eq_matrepe}
\eb=\left(\begin{matrix}
e_0^\natural & -ac\sqrt{g_0} & 0 & \hdots \\
\sqrt{g_0} & e_1^\natural & -acq\sqrt{g_1} & \hdots \\
0 & \sqrt{g_1} & e_2^\natural &  \hdots \\ 
\vdots & \vdots & \vdots & \ddots\\
\end{matrix}\right)\,.
\end{equation}
\end{subequations}
These matrices are similar to those obtained by Sasamoto \cite{Sasamoto:1999aa}.
%\note{If was exactly same as Sasamoto then would be exactly Askey wilson}
 \newcommand{\wmat}{\mathbf{R}}
Clearly the sum, $\wmat= \db + \eb$,   also has a tri-diagonal form,
\begin{equation}\label{eq_edtri}
  \db + \eb =\left(\begin{matrix}
d_0^\natural + e_0^\natural & (1-ac)\sqrt{g_0} & 0 & \hdots \\
(1-bd)\sqrt{g_0} & d_1^\natural + e_1^\natural & (1-acq)\sqrt{g_1} & \hdots \\
0 & (1-bdq)\sqrt{g_1} & d_2^\natural +e_2^\natural &  \hdots \\ 
\vdots & \vdots & \vdots & \ddots\\
\end{matrix}\right)\,,  
\end{equation}
and thus $\wmat$ defines a sequence of orthogonal polynomials,  $\set{T_n(x)}_{n\ge0}$, via the three term recurrence relation obtained from the rows,
\begin{equation}
\wmat_{n,n-1} T_{n-1}+(\wmat_{n,n}-2x) T_n + \wmat_{n,n+1} T_{n+1}=0
\end{equation}
where
\begin{align}
\wmat_{n,n+1}&= (1 - acq^n) \sqrt{g_n}\,,\\
\wmat_{n,n}&=d_n^{\natural} + e_n^{\natural}\,,\\
\wmat_{n,n-1}&= (1 - bdq^{n-1}) \sqrt{g_{n-1}}
\end{align}	
and initial values $T_0=1$ and $ T_{-1}=0$. This three term recurrence is not exactly that for Askey-Wilson polynomials (cf.\ \eqref{eq_awrecrel}) however, the middle coefficient $\wmat_{n,n}$ is the same as the Askey-Wilson  middle recurrence coefficient  and the product of the first and last coefficients $\wmat_{n,n+1} \wmat_{n+1,n}$ is the same as the product  of the first and last coefficients of the the Askey-Wilson recurrence. 
This means that the polynomials $\set{T_n(x)}_{n\ge0}$ have the same moments as the Askey-Wilson polynomials.

%=============================================
\section{Concluding Remarks}
%=============================================

In many cases finding a matrix representation of an algebra (eg.\ the equations of \defref{thm_cw}) can be achieved by the `method of verification': conjecture the matrix elements and then prove they indeed form a representation by showing the requisite matrix sums and products satisfy the algebra. 
This is the apparent `method' used by Littlewood \cite{Littlewood:1933aa} in presenting a representation for the Weyl algebra ($xp-px=1$ -- closely related to \eqref{eq_cw_a})), in the original matrix product ASEP paper \cite{derrida97} (the most general representation is a four parameter representation -- $\alpha$, $\beta$, $\gamma$, $\delta$ with $q=1$) and in the  five parameter ASEP paper  \cite{UCHIYAMA:2003aa}. 

One of the motivations for the previous three parameter paper \cite{brakMore2015} was to try and find   a systematic algebraic method for computing the matrix elements of the representations of $\db$ and $\eb$. The method presented there essentially reduces the determination of the matrix elements of the representation of $\db$ and $\eb$ to the calculation of the determinant of the bi-moment matrix  \eqref{binom} (the determinant gives $\Lam_n$ and hence via \eqref{fmmat} and \eqref{eq_matrefdssd} the matrix elements).

In this paper we have applied the method of \cite{brakMore2015} to the five parameter  case. Going from three to five parameters has a \emph{dramatic} affect on the complexity of the calculations. 
This can be seen at several places. 
It begins  rather subtly, in that the boundary equations  \eqref{eq_cw_a} and \eqref{eq_cw_b}, are now a pair of \emph{coupled} functional equations which \emph{cannot} be simply solved as a pair of `simultaneous linear equations' (which would have made the calculation only a little more complex than the three parameter case) -- the actual method required is   more complex and detailed in the proof of \thmref{bimompascal}.

Another, more  significant  impact, is on the complexity of calculating the determinant of the bi-moment matrix. 
To compute a determinant one usually has the matrix elements explicitly, however, the matrix elements of the bi-moment matrix $\bmm_{i,j}$ are not given explicitly but indirectly via the partial $q$-difference equations stated in \thmref{bimompascal}. 
In the case of the three parameter model \cite{brakMore2015} the boundary matrix elements $\bmm_{i,0}$ and $\bmm_{i,0}$ (required to solve the $\bmm_{i,j}$  partial difference equation) \emph{are} given explicitly (as simple monomials), however in the five parameter model the boundary matrix elements are   given \emph{implicitly} by three term recurrence relations and hence are themselves  related to the value of yet another  set of non-trivial $q$-orthogonal polynomials. 
Thus, if one wanted to use a determinant evaluation method that required explicit expressions for the matrix elements  $\bmm_{i,j}$, one has to solve a partial $q$-difference equation whose boundary values are given implicitly as the values of $q$-orthogonal polynomials. 
Once these two tasks have been achieved one can then  try to compute the determinant. 

To circumvent the difficulty of computing the matrix elements explicitly we turned to  $LDU$ decomposition of the bi-moment matrix. This method somewhat mitigates the task of evaluating the   matrix elements  explicitly by translating the $\bmm_{i,j}$ partial difference equation into partial difference equations for the $L$ and $U$ matrix elements (our primary conjecture -- \eqref{eq_conj}). Given triangular $L$ and $U$  we get a diagonal bi-moment matrix. Once $\bmm$ is diagonalised   the determinant evaluation is straightforward -- see \eqref{eq_bmmdiage}. 

If  $\bmm$ is interpreted as a linear operator, $\bmm: V_1\to V_2$, between two infinite dimensional vectors spaces then clearly the matrix elements of $\bmm$ are determined by the choice of basis for $V_1$ and $V_2$. If the matrix elements of $\bmm$ are defined by \eqref{eq_bmmdef}  then clearly $\bmm$ is  \emph{not} diagonal in the basis implicitly   used for $V_1$ and $V_2$.

This brings us to the primary significance (for this calculation) of the set of  `boundary bases' vectors, $\bra{\hat{W}_n}_{n\ge0}$  and $\ket{\hat{V}_n}_{n\ge0}$  constructed in this paper (see \eqref{eq_orthobasis}). In this basis the  $U$ and $L$ matrices are triangular and hence by choosing the boundary basis   the linear map $\bmm$ has a diagonal matrix representation and hence its determinant becomes a product of the diagonal elements \eqref{eq_bmmdiage}. 
 The `boundary bases' vectors have the further significance in that it is in this basis that the matrix representation of $\db+\eb$ is tri-diagonal and hence defines a three term recurrence relation (see \eqref{eq_edtri}). 
 It is this recurrence relation that is related to the Askey-Wilson orthogonal polynomials. 
 Whilst the  Askey-Wilson polynomials (or their moments) do not appear to have any particular physical significance it is an open question as to whether or not the boundary basis vectors have any physical interpretation.

It would be very interesting to determine if the method of this paper and the new basis it defines has implications (eg.\ new representations) for other ASEP (or similar)  related work associated with finding representations such as the finite representations of    Mallick and  Sandow \cite{Mallick:1997aa},  Sandow \cite{Sandow:1994aa}, the MacDonald and Koornwinder polynomials that appear in Cantini  et.\ al.\ \cite{Cantini:2015aa} and Finn and Vanicat \cite{Finn:2017aa} as well as  the connection to Schur polynomails that appear in Crampe et.\ al. \cite{Crampe:2015aa}.

% The boundary basis vectors defined by the bi-orthogonal polynomials results in a tri-diagonal representation of $\db$, $\eb$ and \mbox{$\db+\eb$} matrices related to the Askey-Wilson orthogonal polynomials. 
% The bi-orthogonal polynomials sequences now satisfy a three recurrence relation (unlike the two term recurrence of the three parameter case) and are thus orthogonal (respectively) to themselves as well as each other. 
% The bi-orthogonal polynomials sequences also generate  the change of basis from the standard basis to the boundary basis (in which the bimoment matrix is diagonal).

%
%=============================================
\section{Acknowledgement} We would like to   thank the Australian Research Council (ARC) and the Centre of Excellence for Mathematics and Statistics of Complex Systems (MASCOS) for financial support and for useful comments from the referees. We would also like to thanks Prof.\ C.\ Krattenthaler, Prof.\ J.\ de Gier, Dr.\ M.\ Wheeler and Dr.\ N.\ Witte for many useful discussions.

%=============================================
\newpage
\bibliographystyle{unsrt}%Used BibTeX style is unsrt
\bibliography{fiveParamBiblio}

\end{document}